\documentclass[aps,pra,reprint]{revtex4-2}

\usepackage[T1]{fontenc}
\usepackage[utf8]{inputenc}
\usepackage{mathptmx}
\usepackage[cal=cm]{mathalfa}
\usepackage{amsmath}
\usepackage{amssymb}
\usepackage{amssymb}
\usepackage{physics}
\usepackage{graphicx}
\usepackage[caption=false]{subfig}
\usepackage{multirow}
\usepackage{array}
\usepackage{makecell}
\usepackage{dcolumn}
\usepackage{xcolor}
\usepackage{comment}
\usepackage{float}
\usepackage{xparse}
\usepackage[colorlinks=true, linkcolor=red, citecolor=blue, urlcolor=blue]{hyperref}
\graphicspath{ {./images1/} }

\def\be{\begin{equation}}
\def\ee{\end{equation}}
\def\bea{\begin{eqnarray}}
\def\eea{\end{eqnarray}}
\def\ben{\begin{equation*}}
\def\een{\end{equation*}}
\def\bean{\begin{eqnarray*}}
\def\eean{\end{eqnarray*}}
\def\bi{\begin{itemize}}
\def\ei{\end{itemize}}
\def\bd{\begin{description}}
\def\ed{\end{description}}
\newcommand{\proj}[1]{\ket{#1}\bra{#1}}

\makeatletter
\@ifundefined{proof}{%
 \newenvironment{proof}[1][Proof]{\par\noindent\textit{#1.}\ }{\hfill$\square$\par}%
}{}
\makeatother
\newtheorem{theorem}{Theorem}
\newtheorem{lem}{Lemma}
\newtheorem{cor}{Corollary}

\newtheorem{defin}{Definition}

\begin{document}
	\title{Inequivalent Quantum Resources from Multipartite State Discrimination}
	\author{Srijani Das}
	\email{dassrijani33@gmail.com}
	\affiliation{Physics and Applied Mathematics Unit, 203 B.T. Road Indian Statistical Institute Kolkata, 700108}
  \author{Tuhin Paul}
	\email{tuhin.paul96@gmail.com}
	\affiliation{Physics and Applied Mathematics Unit, 203 B.T. Road Indian Statistical Institute Kolkata, 700108}
  \author{Manasi Patra}
	\email{manasipatra58@gmail.com}
	\affiliation{Physics and Applied Mathematics Unit, 203 B.T. Road Indian Statistical Institute Kolkata, 700108}
  \author{Atanu Bhunia}
	\email{atanu.bhunia31@gmail.com}
	\affiliation{Physics and Applied Mathematics Unit, 203 B.T. Road Indian Statistical Institute Kolkata, 700108}
  \author{Ramij Rahaman}
	\email{ramijrahaman@isical.ac.in}
	\affiliation{Physics and Applied Mathematics Unit, 203 B.T. Road Indian Statistical Institute Kolkata, 700108}
	
\begin{abstract}

Any two orthogonal quantum states can be perfectly distinguished via local operations and classical communication (LOCC), yet whether all orthogonal multipartite pairs incur an identical operational cost remains an open question. Here, we address the resource consumption required for local state discrimination across multipartite pair of states. We demonstrate that distinct pairs of orthogonal multipartite maximally entangled states require unequal amounts of quantum resources to achieve local distinguishability. Extending our framework beyond individual pairs, we analyze larger sets of both locally distinguishable and indistinguishable multipartite orthogonal states, revealing a nontrivial hierarchy in their entanglement consumption. These results show that the resource cost of LOCC discrimination is state-dependent, providing insight into the fine-grained structure of multipartite nonlocality.
\end{abstract}
  
\maketitle
\section{Introduction} 

A central goal of quantum information science is to harness nonclassical resources to surpass classical processing bounds~\cite{ChitambarResourceTheory2019}. Among these, entanglement is the key resource underpins foundational protocols including superdense coding~\cite{BennettWiesner1992}, teleportation~\cite{Bennett1993}, and quantum key distribution~\cite{Ekert1991}. While quantum correlations are typically consumed during such tasks, protocols that preserve or regenerate entanglement uncover a richer landscape of operational advantages. Notably, Jonathan and Plenio~\cite{JonathanPlenio1999} demonstrated that local operations and classical communication (LOCC) can be catalyzed by an auxiliary entangled state that is returned strictly intact~\cite{Ganardi2025}. This raises two foundational questions: can a protocol achieve a net resource gain relative to the assistance supplied, and what constitutes the minimal resource cost for a given task?

These questions bear directly on the local distinguishability of orthogonal states, an operational manifestation of quantum nonlocality distinct from Bell nonlocality~\cite{BennettPB1999,Walgate2000}. When spatially separated parties cannot perfectly identify an ensemble of mutually orthogonal states via LOCC, the ensemble exhibits nonlocality that underpins quantum data hiding~\cite{terhal58,divincenzo580} and secret sharing~\cite{rahaman330,markham309}. Although any two orthogonal pure states are perfectly distinguishable by LOCC~\cite{Walgate2000}, larger ensembles frequently exhibit local indistinguishability even without entanglement, a phenomenon termed ``nonlocality without entanglement''~\cite{BennettPB1999,Cohen2008,somsubhro2018,Halder2018,Gupta2023Hierarchical,bhunia2020,bhunia2023,bhunia2022,bhunia2026,bhuniaubb2024,subrata2024,bhunia2025,yu2012}. 

Conventional state discrimination protocols, however, generally consume or degrade the intrinsic quantum correlations of the target states. This motivates a sharper operational inquiry: can the classical information encoded in an entangled ensemble be extracted via LOCC while leaving the target states intact? While auxiliary entanglement can enable discrimination below the resource cost of teleportation~\cite{Cohen2008}, the ancillae are typically consumed. Cohen~\cite{Cohen2007} addressed the stricter regime where the target bipartite entanglement is conserved, revealing an inherent trade-off between consumed and protected resources. More recently, Halder \emph{et al.}~\cite{Halder2026} established cardinality bounds for discrimination assisted by exact entanglement batteries and catalysts. Nevertheless, determining the minimal auxiliary entanglement required to discriminate orthogonal states while preserving their intrinsic multipartite entanglement remains an open problem.

In this Letter, we address this challenge across bipartite and multipartite regimes. For bipartite systems, we establish that net resource extraction is fundamentally precluded when discriminating Bell states under LOCC. In the multipartite setting, however, an asymmetry emerges: two orthogonal $n$-qubit Greenberger-Horne-Zeilinger (GHZ) \cite{Greenberger1989} states with nonzero relative Hamming distance can be perfectly distinguished via LOCC assisted by a single two-qubit Bell pair, leaving the target GHZ state strictly undisturbed. This catalytic advantage vanishes at zero Hamming distance, where auxiliary entanglement provides no operational benefit for nondestructive discrimination. Extending this framework to larger ensembles of maximally and nonmaximally entangled GHZ states, we prove that any set possessing pairwise nonzero Hamming distances is perfectly distinguishable, whereas specific triples remain strictly indistinguishable under LOCC. Finally, we provide a complete classification of orthogonal GHZ ensembles that admit LOCC discrimination under minimal auxiliary ebits, delineating the exact boundaries where target correlations can be preserved.
\section{Resource extraction in bipartite LOCC discrimination}
{\textit{Quantum resources under LOCC}}- Within the resource theory of entanglement, the free operations are local operations and classical communication (LOCC). The set of free states $\mathcal{F}$ consists of all separable states, $\sigma_{\mathrm{sep}} = \sum_i p_i \rho_A^{(i)} \otimes \rho_B^{(i)}$,
where $\{p_i\}$ is a probability distribution and $\rho_A^{(i)}, \rho_B^{(i)}$ are local density operators. States $\rho \notin \mathcal{F}$ are entangled and constitute the quantum resource. A functional $\mathcal{R}$ qualifies as an entanglement monotone (resource quantifier) if it satisfies:
(i) \emph{Faithfulness}: $\mathcal{R}(\rho) \ge 0$, with $\mathcal{R}(\rho) = 0$ if and only if $\rho \in \mathcal{F}$;
(ii) \emph{Monotonicity}: $\mathcal{R}(\Lambda(\rho)) \le \mathcal{R}(\rho)$ for any deterministic operation $\Lambda \in \mathrm{LOCC}$; and
(iii) \emph{Ensemble monotonicity}: $\sum_k q_k \mathcal{R}(\rho_k) \le \mathcal{R}(\rho)$ under selective LOCC operations yielding outcomes with probabilities $q_k$ and normalized post-measurement states $\rho_k$.

\begin{defin}[Catalytic LOCC]
A catalytic LOCC (C-LOCC) process is an operation $\Lambda \in \mathrm{LOCC}$ acting jointly on a primary system $S$ and an auxiliary catalyst $C$ initialized in state $\sigma_C$. With respect to a given resource measure $\mathcal{R}$, the protocol is catalytic if the final marginal state of the catalyst, $\sigma'_C = \mathrm{Tr}_S[\Lambda(\rho_S \otimes \sigma_C)]$, satisfies $\mathcal{R}(\sigma'_C) \ge \mathcal{R}(\sigma_C)$. Exact catalysis corresponds to $\sigma'_C = \sigma_C$.

If the auxiliary resource is degraded or consumed, $\mathcal{R}(\sigma'_C) < \mathcal{R}(\sigma_C)$, we refer to the transformation as resource-assisted LOCC (RA-LOCC).
\end{defin}
\begin{defin}[C-LOCC Resource extraction]
A C-LOCC process $\Lambda$ is termed \emph{resource-extracting LOCC} (RE-LOCC) with respect to a measure $\mathcal{R}$ if the post-processing marginal state $\sigma'_C = \mathrm{Tr}_S[\Lambda(\rho_S \otimes \sigma_C)]$ strictly exceeds the initial catalytic resource, $\mathcal{R}(\sigma'_C) > \mathcal{R}(\sigma_C)$.
\end{defin}

We now examine the simplest setting to expose the resource cost of local discrimination. For two-qubit Bell states, perfect LOCC discrimination necessarily consumes the entanglement contained in the state.

\begin{theorem}\label{th_GHZ_1}
No quantum resource can be extracted under LOCC (or C-LOCC) discrimination of two Bell states auxiliary to the mutual information identifying the input state.
\end{theorem}

\begin{proof}
By local-unitary equivalence, it suffices to consider the pair of Bell states,
\be\label{Bell_2}
|\Phi^\pm\rangle=\tfrac{1}{\sqrt{2}}\left(|00\rangle\pm|11\rangle\right),
\ee
which can be perfectly distinguished by LOCC~\cite{Walgate2000}. For an equiprobable ensemble, the state prior to discrimination is
\[
\rho=\tfrac{1}{2}\left(|\Phi^+\rangle\langle\Phi^+|
+|\Phi^-\rangle\langle\Phi^-|\right)
=\tfrac{1}{2}\left(
|00\rangle\langle00|+
|11\rangle\langle11|
\right),
\]
which is manifestly separable. Hence, any quantum resource quantified by a monotone under LOCC has vanishing value on \(\rho\) and cannot be generated by the subsequent LOCC or C-LOCC discrimination. The nontrivial output of the protocol is therefore confined to the classical information identifying the input Bell state.
\end{proof}

A trivial extension of the Theorem \ref{th_GHZ_1} can be stated as 
\begin{cor}
    No quantum resource can be extracted under LOCC (or C-LOCC) discrimination of Bell states auxiliary to the mutual information identifying the input state.
\end{cor}

\begin{proof} The two-Bell-state case follows directly from Theorem~\ref{th_GHZ_1}, leaving only ensembles of three or four Bell states to be addressed. Suppose that an LOCC or C-LOCC protocol could extract resource during the discrimination of three or four Bell states. The same protocol would then extract at least as much resource when discriminating any two-state subensemble, in direct contradiction with Theorem~\ref{th_GHZ_1}. Each case can also be proved independently. \end{proof}

\section{Resource extraction in multipartite LOCC discrimination}

We now extend the analysis to general \(n\)-qubit entangled states. To quantify the structural separation between their computational-basis components, we use the Hamming distance.

{\setlength{\abovedisplayskip}{3pt}\setlength{\belowdisplayskip}{3pt}%
 \setlength{\abovedisplayshortskip}{1pt}\setlength{\belowdisplayshortskip}{1pt}%
\begin{defin}[Hamming distance]
For bit strings $\mathbf{x},\mathbf{y}\in\{0,1\}^n$, $d_H(\mathbf{x},\mathbf{y})=\sum_{k=1}^{n}(x_k\oplus y_k)$ is the number of positions in which they differ.
\end{defin}
We use $d_H$ to characterize the structural separation between entangled states through their computational-basis strings.

\begin{defin}[Hamming distance of a set of states]
Consider two normalized $n$-qubit states, \ben \displaystyle\vert{}\psi\rangle = \sum_{\mathbf{x} \in \operatorname{\mathrm{supp}}(\psi)} \lambda_{\mathbf{x}} \vert{}\mathbf{x}\rangle \text{ and } \vert{}\phi\rangle = \sum_{\mathbf{y} \in \operatorname{\mathrm{supp}}(\phi)} \gamma_{\mathbf{y}} \vert{}\mathbf{y}\rangle,\een where $\mathbf{x}, \mathbf{y} \in \{0,1\}^n$ and the sums run over all basis configurations with nonvanishing amplitudes. We define the Hamming distance between $\vert{}\psi\rangle$ and $\vert{}\phi\rangle$ as the minimum bitwise distance across their computational supports,
\begin{equation*}
d_H(|\psi\rangle, |\phi\rangle) \equiv \min_{\substack{\mathbf{x} \in \operatorname{\mathrm{supp}}(\psi) \ \mathbf{y} \in \operatorname{\mathrm{supp}}(\phi)}} d_H(\mathbf{x}, \mathbf{y}).
\end{equation*}
Extending this metric to an ensemble of $k$ distinct states, $\mathcal{S} = \{\vert{}\psi_i\rangle\}_{i=1}^k$, the Hamming distance of the set corresponds to the closest pair:
\ben
d_H(\mathcal{S}) \equiv \min_{i \neq j} \left\{d_H(|\psi_i\rangle, |\psi_j\rangle): \textbf{ for all } |\psi_i\rangle, |\psi_j\rangle\in \mathcal{S}\right\}.\een

\end{defin}}

\subsection{Entanglement non-preserving state discrimination}
The Hamming distance introduced here establishes an operational criterion for both preserving and non-preserving state discrimination, quantifying the degree to which post-measurement states remain preserved. In a state preserving discrimination, allowing the states to be identified while leaving the input system completely undisturbed. This metric naturally interfaces with the operational costs of quantum resource distillation under LOCC. In this letter, we focus on multipartite orthogonal sets, with particular emphasis on GHZ states. A complete orthonormal basis of $n$-GHZ states can be parametrized up to an irrelevant global phase as
\be\label{GHZ_basis}
\ket{\psi_n({q,\mathbf{r}})}=\dfrac{1}{\sqrt{2}}\left(\ket{1}^{\otimes \mathbf{r}}\ket{0}^{\otimes \bar{\mathbf{r}}}+(-1)^q\ket{0}^{\otimes \mathbf{r}}\ket{1}^{\otimes \bar{\mathbf{r}}}\right),
\ee
where $q\in\{0,1\}$, $\mathbf{r}\subseteq \{1,2,\dots,n-1\}$ and $\bar{\mathbf{r}}=\{1,2,\dots,n\}\smallsetminus \mathbf{r}$ denotes the complement of $\mathbf{r}$. For $\mathbf{r} = \emptyset$ (empty set), Eq.~\eqref{GHZ_basis} reduces to the standard canonical GHZ state, $\ket{\psi_n(q,\emptyset)} = (\ket{0}^{\otimes n} + (-1)^q \ket{1}^{\otimes n})/\sqrt{2}$. The effective overlap between any two basis states $\ket{\psi_n(q,\mathbf{r})}$ and $\ket{\psi_n(p,\mathbf{m})}$ is characterized by the Hamming distance
\begin{equation*}
d_H\left(\ket{\psi_n(q,\mathbf{r})}, \ket{\psi_n(p,\mathbf{m})}\right) \equiv \left|\left(\mathbf{r} \cup \mathbf{m}\right)\smallsetminus \left(\mathbf{r} \cap \mathbf{m}\right)\right|,
\end{equation*}
defined here via the cardinality of the symmetric difference of the active index sets.

\begin{theorem}\label{th_GHZ_2}
Under LOCC (or C-LOCC), discriminating any pair of orthogonal $n$-qubit GHZ states separated by zero relative Hamming distance cannot yield any net quantum resource beyond the classical information identifying the state. 
\end{theorem}

\begin{proof} Any pair of orthogonal $n$-qubit GHZ states sharing zero relative Hamming distance is equivalent, up to local unitaries, to the canonical phase-split pair
\begin{equation*}
|\psi_n(p, \emptyset)\rangle = \frac{1}{\sqrt{2}} \left( |0\rangle^{\otimes n} + (-1)^p |1\rangle^{\otimes n} \right), \quad p \in \{0,1\}.
\end{equation*}
Since two orthogonal pure states are deterministically distinguishable via LOCC~\cite{Walgate2000}, here achieved by having each party measure in the local $\sigma_x$ basis and combining outcomes via classical communication, the protocol unavoidably destroys the target state.

To evaluate potential quantum resource extraction, consider the a priori state prior to measurement. For an equiprobable ensemble, the initial state is given by the statistical mixture
\begin{equation*}
\rho_1 = \frac{1}{2} \sum_{p=0}^1 |\psi_n(p, \emptyset)\rangle \langle \psi_n(p, \emptyset)| = \frac{1}{2} |0\rangle\langle 0|^{\otimes n} + \frac{1}{2} |1\rangle\langle 1|^{\otimes n}.
\end{equation*}
The ensemble average $\rho_1$ is a classical mixture of fully product states; hence, it is entirely separable across every partition and possesses zero genuine multipartite entanglement. Because LOCC (as well as catalytic LOCC) operations cannot generate entanglement from separable states, no quantum resource can be extracted or converted from the ensemble. The process yields only classical mutual information corresponding to the identity of $p$, completing the proof. 
\end{proof}

\begin{cor} Under LOCC (or C-LOCC), discriminating any ensemble of mutually orthogonal $n$-qubit GHZ states that contains a pair with zero Hamming distance yields no quantum resource. \end{cor}

The result follows directly from Theorem~\ref{th_GHZ_2}: if distillable resources could be extracted during the discrimination of a given ensemble, they could likewise be extracted when discriminating any of its subensembles, including a zero-distance orthogonal pair.

\begin{theorem}\label{th_2nGHZ}
  Any ensemble of mutually orthogonal $n$-qubit GHZ states possessing pairwise nonzero Hamming distances can be deterministically distinguished via LOCC. 
\end{theorem}

\begin{proof} The maximal cardinality of such an ensemble is $2^{n-1}$; establishing deterministic LOCC distinguishability for this maximal set immediately implies the result for any subset. Up to local unitaries and global phases, any maximal ensemble $\mathcal{S}_+$ can be parameterized as
\begin{equation*}
\mathcal{S}_+ = \left\{ |\psi_n(0, \mathbf{r})\rangle : \mathbf{r} \subseteq \{1, 2, \dots, n-1\} \right\},
\end{equation*}
where $|\psi_n(0, \mathbf{r})\rangle$ is defined in Eq.~\eqref{GHZ_basis}, and distinct subsets $\mathbf{r} \neq \mathbf{m}$ correspond to orthogonal states.

To distinguish the states deterministically, each party measures their local qubit independently in the computational ($\sigma_z$) basis. For an unknown state $|\psi_n(\mathbf{r})\rangle$, the measurement outcomes within $\mathbf{r}$ are perfectly correlated, as are those within $\bar{\mathbf{r}}$, while an exact anticorrelation emerges between the two complementary sets. Broadcast of these classical measurement results over classical channels unambiguously identifies $\mathbf{r}$, completing the deterministic LOCC discrimination protocol.
\end{proof}

To generalize these results to nonmaximally entangled GHZ states, we express the nonmaximal GHZ basis, up to an overall phase, as
\begin{equation}
\label{GHZ_basis_nonmax}
\ket{\psi_n^{\alpha^{\mathbf{r}}}(q,\mathbf{r})} = \cos\alpha_q^{\mathbf{r}} \ket{1}^{\otimes \mathbf{r}} \ket{0}^{\otimes \bar{\mathbf{r}}} + \sin\alpha_q^{\mathbf{r}} \ket{0}^{\otimes \mathbf{r}} \ket{1}^{\otimes \bar{\mathbf{r}}},
\end{equation}
where $\alpha_q^{\mathbf{r}} \equiv \alpha^{\mathbf{r}} - q(\pi/2)$ with $\alpha^{\mathbf{r}} \in (0, \pi/2)$, $q \in \{0,1\}$, and $\bar{\mathbf{r}} = \{1, \dots, n\} \setminus \mathbf{r}$ for any subset $\mathbf{r} \subseteq \{1, \dots, n-1\}$. Theorem~\ref{th_2nGHZ} extends directly to nonmaximally entangled GHZ states:

\begin{cor}
Any set of mutually orthogonal $n$-qubit, maximally or nonmaximally entangled, GHZ states of the form Eq.~\eqref{GHZ_basis_nonmax} with nonzero pairwise Hamming distances is perfectly distinguishable by LOCC.
\end{cor}

\begin{theorem}\label{th_3GHZ}
No ensemble of three mutually orthogonal $n$-qubit GHZ states possessing zero Hamming distance can be deterministically distinguished via LOCC.
\end{theorem}

\begin{proof}
Up to local unitaries, any triplet containing a zero-Hamming-distance pair can be mapped to the canonical form [cf.~Eq.~\eqref{GHZ_basis}]:
\begin{equation}
\label{GHZ_3}
\begin{aligned}
|\psi_n(q,\emptyset)\rangle &= \frac{1}{\sqrt{2}}\left(|0\rangle^{\otimes n} + (-1)^q|1\rangle^{\otimes n}\right), ~q=0,1, \text{ and}\\
|\psi_n(0,\mathbf{r})\rangle &= \frac{1}{\sqrt{2}}\left(|1\rangle_{\mathbf{r}}|0\rangle_{\bar{\mathbf{r}}} + |0\rangle_{\mathbf{r}}|1\rangle_{\bar{\mathbf{r}}}\right),
\end{aligned}
\end{equation}
with $\emptyset \neq \mathbf{r} \subset \{1,\dots,n\}$ and $\bar{\mathbf{r}}$ denoting its complement. 

Now consider the bipartite cut $\mathbf{r}|\bar{\mathbf{r}}$, where parties within $\mathbf{r}$ and $\bar{\mathbf{r}}$ may perform joint operations within their respective local laboratories. Across this cut, the three states in Eq.~\eqref{GHZ_3} reduce to three orthogonal two-qubit Bell states: $|\Phi^+\rangle$, $|\Phi^-\rangle$, and $|\Psi^+\rangle$. Because any LOCC protocol under the original $n$-partite partition remains valid under this bipartite coarsening, the indistinguishability of three Bell states~\cite{ghosh2001distinguishability} strictly rules out deterministic LOCC discrimination of the multipartite ensemble. 
\end{proof}
\begin{cor}
No ensemble of $2^{n-1} + 1$ mutually orthogonal $n$-qubit GHZ states can be deterministically distinguished by LOCC.
\end{cor}
This follows directly from Theorem~\ref{th_3GHZ}. 
\begin{theorem}\label{th_GHZ_any} 
Any ensemble $\Omega$ of $k > 2$ mutually orthogonal $n$-qubit GHZ states containing at least one pair with zero Hamming distance cannot be deterministically distinguished via LOCC. While assisted discrimination is achievable using $\lfloor \log_2 k \rfloor$ ebits of shared entanglement, deterministic preservation of the target GHZ state or resource extraction remains strictly forbidden.
\end{theorem}
\begin{proof}
The unassisted LOCC indistinguishability follows directly from Theorem~\ref{th_3GHZ}. Furthermore, the presence of a pair with zero Hamming distance implies that $\Omega$ contains both phase counterparts, $\{|\psi_n(0,\mathbf{r})\rangle, |\psi_n(1,\mathbf{r})\rangle\}$, for some configuration $\mathbf{r}$. By Theorem~\ref{th_GHZ_2}, deterministically preserving the target state or extracting quantum resources from this pair under LOCC is fundamentally precluded.

To establish assisted distinguishability, we decompose the ensemble as $\Omega = \Omega_+ \cup \Omega_-$, where $\Omega_\pm = \Omega \cap \mathcal{S}_\pm$ and $\mathcal{S}_\pm = \{|\psi_n(\substack{0\\1},\mathbf{r})\rangle\}_{\mathbf{r}\subseteq\{1,2,\dots,n-1\}}$. At each step of the identification, there exist two parties whose joint computational parity ($\sigma_z \otimes \sigma_z$) bisects the candidate set into subsets of size at most $\lfloor |\Omega|/2 \rfloor$. To extract this parity without disturbing the target state, the two parties consume a shared Bell pair $|\Phi^+\rangle$ as an ancilla. Each party applies a local $\textsc{cnot}$ gate controlled by their respective GHZ qubit and targeting their half of the Bell pair (see the Fig. \ref{fig:1}); a subsequent local $\sigma_z \otimes \sigma_z$ measurement on the ancilla reveals the parity partition nondestructively.

Iterating this bisection at most $\lfloor \log_2 k \rfloor$ times uniquely isolates the vector index $\mathbf{r}$. If the candidate set contains only a single state corresponding to $\mathbf{r}$ (i.e., $|\psi_n(1,\mathbf{r})\rangle \notin \Omega_-$), the state is fully identified and remains completely intact. Conversely, if both phase states are present ($|\psi_n(1,\mathbf{r})\rangle \in \Omega_-$), the remaining binary phase index $q \in \{0,1\}$ is resolved by measuring every qubit in the local $\sigma_x$ basis. Denoting the total number of $-1$ measurement outcomes by $m$, the global parity $m \oplus |\mathbf{r}|$ deterministically yields $q$. This final projective step destroys the target state precisely when phase discrimination is required, completing the proof.
\end{proof}
\begin{lem}
Any ensemble of mutually orthogonal $n$-qubit GHZ states, whether maximally or nonmaximally entangled, can be deterministically distinguished via LOCC given two copies of the unknown state.
\end{lem}

\begin{proof}
We establish the protocol explicitly for maximally entangled states; the extension to nonmaximal states follows directly. Up to local unitaries and an overall phase, any ensemble of mutually orthogonal $n$-qubit GHZ states embeds into the complete basis family
\begin{equation*}
\mathcal{S} = \left\{ |\psi_n(q,\mathbf{r})\rangle : q \in \{0,1\}, \, \mathbf{r} \subseteq \{1, 2, \dots, n-1\} \right\},
\end{equation*}
where $\bar{\mathbf{r}}$ denotes the complement of $\mathbf{r}$. Because deterministically distinguishing $\mathcal{S}$ guarantees the distinguishability of any subset, it suffices to discriminate the parent set.

The protocol proceeds sequentially across the two copies:
First, each party measures their local qubit of the first copy in the computational ($\sigma_z$) basis. Broadcasting these single-bit outcomes identifies the spatial pattern of correlations and anticorrelations, uniquely resolving the subset index $\mathbf{r}$ via Theorem~\ref{th_2nGHZ}.

Second, each party measures their corresponding qubit of the second copy in the local $\sigma_x$ basis and broadcasts their outcome $m_j \in \{\pm 1\}$. Because $|\psi_n(q,\mathbf{r})\rangle$ is a collective eigenstate satisfying $\sigma_x^{\otimes n} |\psi_n(q,\mathbf{r})\rangle = (-1)^q |\psi_n(q,\mathbf{r})\rangle$, the total count $k$ of $-1$ outcomes directly yields the phase index via $q \equiv k \pmod 2$. Consequently, both $(\mathbf{r}, q)$ are determined unambiguously, completing the proof.
\end{proof}

\begin{theorem}
The four $n$-qubit nonmximally GHZ states
\begin{equation}
\label{eq:four_states}
\begin{aligned}
\ket{\psi_n^\alpha(q, 0)} &= \cos\alpha_q \ket{00}\ket{\mathbf{x}} \pm \sin\alpha_q \ket{11}\ket{\bar{\mathbf{x}}}, \\
\ket{\psi_n^\beta(q, 1)} &= \cos\beta_q \ket{01}\ket{\bar{\mathbf{x}}} \pm \sin\beta_q \ket{10}\ket{\mathbf{x}},
\end{aligned}
\end{equation}
parametrized by $\alpha_q \equiv \alpha - q(\pi/2)$ and $\beta_q \equiv \beta - q(\pi/2)$ with $\alpha, \beta \in (0, \pi/2)$, $q \in \{0,1\}$, and $\mathbf{x} \in \{0,1\}^{n-2}$ with bitwise complement $\bar{\mathbf{x}}$, can be deterministically distinguished via LOCC assisted by a single shared ebit.
\end{theorem}

\begin{proof} 
Let Alice and Bob hold qubits $1$ and $2$, respectively, while the remaining $n-2$ qubits associated with $\mathbf{x}$ are distributed among the other parties. Alice and Bob initially share an auxiliary Bell pair $|\Phi^+\rangle_{A'B'} = \frac{1}{\sqrt{2}}(|00\rangle + |11\rangle)_{A'B'}$. 

Each applies a local $\textsc{cnot}$ gate controlled by their respective target qubit and acting on their half of the Bell pair [Fig.~\ref{fig:1}]. Because computational-basis control operations leave the control registers invariant, the target state remains undisturbed, while the ancilla transforms conditionally on the joint parity of qubits $1$ and $2$. A subsequent local $\sigma_z \otimes \sigma_z$ measurement on the ancilla extracts this parity nondestructively: correlated outcomes ($00$ or $11$) isolate the branch $|\psi_n^\alpha(q, 0)\rangle$, whereas anticorrelated outcomes ($01$ or $10$) yield $|\psi_n^\beta(q, 1)\rangle$.

With the candidate subspace identified and the $n$-qubit state intact, the remaining binary index $q \in \{0,1\}$ is deterministically resolved using the general LOCC protocol for two orthogonal pure states~\cite{Walgate2000}-namely, by Alice measuring in a rotated basis adapted to $\alpha$ (or $\beta$) and coordinating with the other parties via classical communication. 
\end{proof}

\subsection{Entanglement preserving state discrimination}
\begin{theorem}\label{th_2GHZ_nonzero}
For any pair of orthogonal $n$-qubit GHZ states with nonzero Hamming distance, one Bell state is both necessary and sufficient for LOCC discrimination without destroying the target GHZ correlation.
\end{theorem}
\begin{proof}  
Following the parameterization in Eq.~\eqref{GHZ_basis}, any two orthogonal $n$-qubit GHZ states differing by a nonzero Hamming distance can be expressed as \(\displaystyle \{|\psi_n(p,\mathbf{r})\rangle, |\psi_n(q,\mathbf{m})\rangle\}\) with $\mathbf{r}\neq \mathbf{m}$ and $p,q\in\{0,1\}$, $\mathbf{r},\mathbf{m}\subseteq \{1,2,\dots,n-1\}$. 

Although these states can be deterministically distinguished via standard LOCC by consuming the target entanglement~\cite{Walgate2000}, non-destructive discrimination requires an auxiliary resource.

To discriminate the states while preserving the target GHZ correlation, the parties first apply local Pauli rotations $\sigma_z$ to all qubits $i \in \mathbf{r}$. Up to an irrelevant global phase, this maps the given ensemble in to
\begin{equation}
\label{GHZ_Hr2}
\begin{aligned}
\left(\sigma_z^{\mathbf{r}}\otimes \mathbb{I}_2^{\otimes\bar{\mathbf{r}}}\right)|\psi_n(p,\mathbf{r})\rangle &\equiv |\psi_n(p\oplus |\mathbf{r}|,\emptyset)\rangle, \\
\left(\sigma_z^{\mathbf{r}}\otimes \mathbb{I}_2^{\otimes\bar{\mathbf{r}}}\right)|\psi_n(q,\mathbf{m})\rangle &\equiv |\psi_n(q\oplus |\mathbf{r}|,\mathbf{s})\rangle,
\end{aligned}
\end{equation}
where $\mathbf{s}= \mathbf{r} \mathbin{\Delta} \mathbf{m} = (\mathbf{r} \cup \mathbf{m})\setminus (\mathbf{r} \cap \mathbf{m}) \neq \emptyset$ denotes the symmetric difference between $\mathbf{r}$ and $\mathbf{m}$.

Next, choose any party $A_j \in \mathbf{s}$ and any party $A_k \in \bar{\mathbf{s}}$ to share an auxiliary Bell state $|\Phi^+\rangle_{A_j^\prime A_k^\prime} = (|00\rangle + |11\rangle)/\sqrt{2}$ (the circuit is shown in Fig.~\ref{fig:1}).

\begin{figure}[H]
  \centering
  \includegraphics[width=0.8\linewidth]{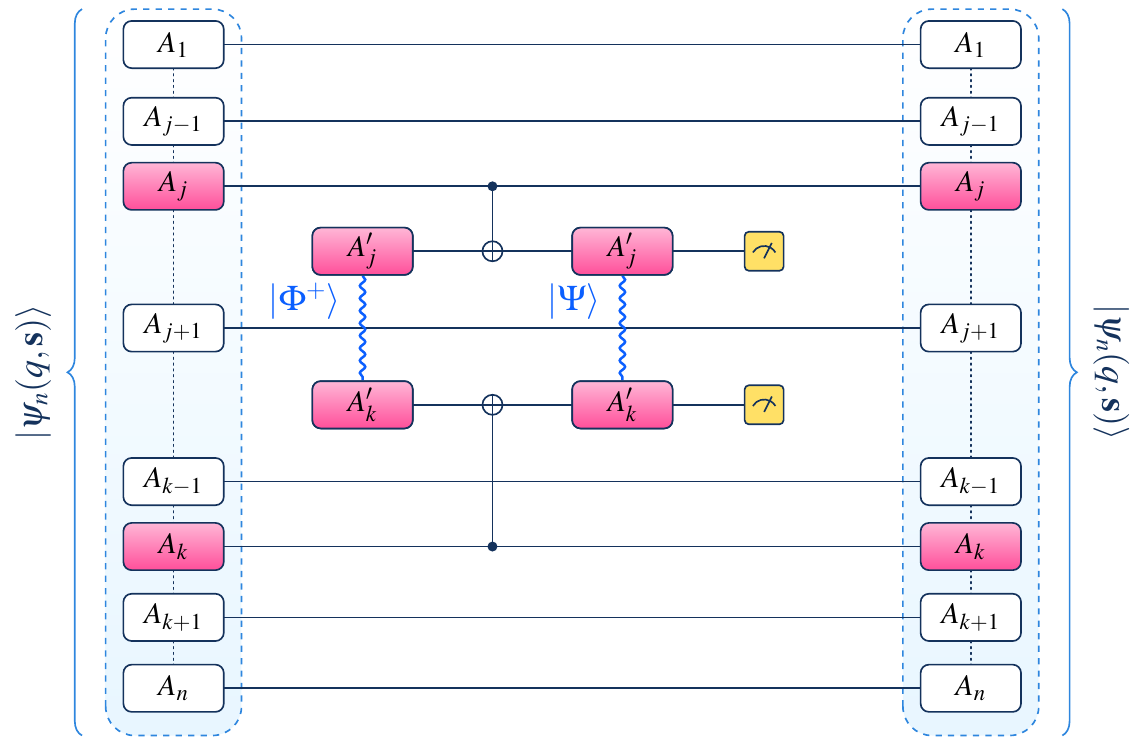}
  \caption{Circuit for nondestructive GHZ-state discrimination using an ancillary Bell state $|\Phi^+\rangle_{A_j' A_k'}$ shared between nodes $A_j (j\in \mathbf{s})$ and $A_k (k\in \bar{\mathbf{s}})$ (pink). Controlled-NOT gates, with the system qubits as controls, are applied from $A_j$ to $A_j'$ and from $A_k$ to $A_k'$. The ancilla remains in $|\Phi^{+}\rangle$ for even parity of the two system qubits and transforms to $|\Psi^{+}\rangle$ for odd parity. Computational-basis measurements of the ancillas therefore reveal the parity without disturbing the input GHZ state, which is recovered unchanged at the output.}
  \label{fig:1}
\end{figure}

Parties $A_j$ and $A_k$ each perform a local controlled-\textsc{not} ($\mathrm{CNOT}$) gate, setting their respective unknown GHZ qubit as the control $c$ and their Bell qubit as the target $t$, where $\mathrm{CNOT}|c\rangle|t\rangle = |c\rangle|t\oplus c\rangle$. Labeling the remaining $n-2$ parties as $E$, the joint state transforms as

  \bean
  &\ket{\psi_n(p\oplus |\mathbf{r}|,\emptyset)}_{A_j^cA_k^cE}\otimes|\Phi^+\rangle_{{A^\prime_j}^t{A^\prime_k}^t}&\\&\xrightarrow[\mathrm{CNOT}_{A_k^c{A^\prime_k}^t} ]{\mathrm{CNOT}_{A_j^c{A_j^\prime}^t}} \ket{\psi_n(p\oplus |\mathbf{r}|,\emptyset)}_{A_j^cA_k^cE}\otimes|\Phi^+\rangle_{{A^\prime_j}^t{A^\prime_k}^t}&\\
 &\mathrm{and}& \\
&\ket{\psi_n(q\oplus |\mathbf{r}|,\mathbf{s})}_{A_j^cA_k^cE}\otimes|\Phi^+\rangle_{{A^\prime_j}^t{A^\prime_k}^t} &\\&\xrightarrow[\mathrm{CNOT}_{A^c_k{A^\prime_k}^t} ]{\mathrm{CNOT}_{A_j^c{A_j^\prime}^t}} \ket{\psi_n(q\oplus |\mathbf{r}|,\mathbf{s})}_{A_j^cA_k^cE}\otimes|\Psi^+\rangle_{{A^\prime_j}^t{A^\prime_k}^t},&
\eean

where $|\Psi^+\rangle = (|01\rangle + |10\rangle)/\sqrt{2}$.

Parties $A$ and $B$ then measure their auxiliary Bell qubits $A_t$ and $B_t$ in the computational ($\sigma_z$) basis and broadcast their measurement outcomes. Correlated outcomes ($00$ or $11$) indicate the target state was $|\psi(p,\mathbf{r})\rangle$, whereas anticorrelated outcomes ($01$ or $10$) indicate $|\psi(q,\mathbf{m})\rangle$. In either case, the target GHZ state is completely unentangled from the auxiliary qubits and remains undisturbed up to the initial local rotation $\sigma_z^{\mathbf{r}}$, which is subsequently inverted. The states are thus distinguished deterministically without destroying the shared GHZ correlation. 

To establish necessity, recall that the given pair of states \(\displaystyle \{|\psi_n(p,\mathbf{r})\rangle, |\psi_n(q,\mathbf{m})\rangle\}\) are mapped via local unitaries to
{\setlength{\abovedisplayskip}{4pt}\setlength{\belowdisplayskip}{4pt}%
 \setlength{\abovedisplayshortskip}{2pt}\setlength{\belowdisplayshortskip}{2pt}%
\begin{align}
\label{GHZ_Hr3}
|\psi_n(q\oplus |\mathbf{r}|,\mathbf{s})\rangle &= \tfrac{1}{\sqrt{2}}\big(|1\rangle^{\otimes \mathbf{s}}|0\rangle^{\otimes \bar{\mathbf{s}}} + (-1)^{q\oplus |\mathbf{r}|}|0\rangle^{\otimes \mathbf{s}}|1\rangle^{\otimes \bar{\mathbf{s}}}\big)\nonumber\\
|\psi_n(p\oplus |\mathbf{r}|,\emptyset)\rangle &= \tfrac{1}{\sqrt{2}}\big(|0\rangle^{\otimes n} + (-1)^{p\oplus|\mathbf{r}|}|1\rangle^{\otimes n}\big), 
\end{align}

Consider the bipartite partition $A|B$, with the parties $\mathbf{s}$ grouped in laboratory $A$ and the complement $\bar{\mathbf{s}}$ in laboratory $B$, and LOCC allowed across the cut. In a suitable local basis, the states in Eq.~\eqref{GHZ_Hr3} reduce across this partition to the pair of Bell states
\begin{equation*}
|\Psi^{+}\rangle_{AB} = \tfrac{1}{\sqrt{2}}\big(|01\rangle + |10\rangle\big), \qquad
|\Phi^{+}\rangle_{AB} = \tfrac{1}{\sqrt{2}}\big(|00\rangle + |11\rangle\big),
\end{equation*}}
which is again local unitarily equivalent to the pair in Eq.~\eqref{Bell_2}. Distinguishing this pair without consuming entanglement is precluded across this cut by Theorem~\ref{th_GHZ_1}. Consequently, non-destructive discrimination of the target GHZ states under LOCC strictly requires an auxiliary entanglement resource of at least one ebit (a shared Bell state).
 
\end{proof}

The sufficiency argument in Theorem~\ref{th_2GHZ_nonzero} readily generalizes to arbitrary pairs of $n$-qubit nonmaximally entangled GHZ states of the form in Eq.~\eqref{GHZ_basis_nonmax} with nonzero Hamming distance, yielding the following result:

\begin{cor} Any two $n$-qubit nonmaximally entangled GHZ states with nonzero Hamming distance can be perfectly distinguished by LOCC assisted by a single shared ebit, while preserving the input GHZ state.\end{cor}

\begin{theorem}\label{th_GHZ_nonzero}
Any ensemble ${\Gamma}$ of $k \geq 2$ mutually orthogonal $n$-qubit GHZ states possessing pairwise non-zero Hamming distance can be deterministically distinguished via LOCC assisted by $\lceil\log_2 k\rceil$ ancillary ebits, leaving the unknown target state completely undisturbed.
\end{theorem}

\begin{proof} 
The protocol follows a recursive bisection strategy similar to that in Theorem~\ref{th_GHZ_any}. At each stage, there exist two parties whose joint $\sigma_z$ parity correlated ($\{\ket{00}, \ket{11}\}$) or anticorrelated ($\{\ket{01}, \ket{10}\}$) partitions the candidate states into subsets of size at most $\lceil |\Gamma|/2 \rceil$. 

To extract this parity non-destructively, the two parties consume a shared Bell pair $\ket{\Phi^+}_{A_1 A_2}$ as an ancilla. Each party applies a local $\textsc{cnot}$ gate controlled by their respective GHZ qubit and targeting their half of the ancilla Bell pair. A subsequent local $\sigma_z \otimes \sigma_z$ measurement on the ancilla reveals the parity partition without disturbing the underlying GHZ state.

Iterating this bisection at most $\lceil \log_2{k} \rceil$ times deterministically isolates the vector index $\mathbf{r}$ and consequently the identity of the unknown GHZ state without disturbing it. \end{proof}

\section{Inequivalent Resources in LOCC Discrimination}
\label{sec:multipartite_entanglement}

Denoting the input states in Theorems~\ref{th_GHZ_2} and \ref{th_2GHZ_nonzero} by $\rho_1$ and $\rho_2$, respectively, we have
\be\label{mix_state_1} \begin{aligned}
  \rho_1 &= \tfrac{1}{2} \Big( \proj{\psi_n(0,\emptyset)} + \proj{\psi_n(1,\emptyset)} \Big) \text{ and} \\
  \rho_2 &= \tfrac{1}{2} \Big( \proj{\psi_n(0,\emptyset)} + \proj{\psi_n(0,\mathbf{r})} \Big),
\end{aligned}\ee
where $\mathbf{r} \subseteq \{1, 2, \dots, n-1\}$ and $\mathbf{r}\neq \emptyset$.

\begin{lem}
Neither $\rho_1$ nor $\rho_2$ defined in Eq.~\eqref{mix_state_1} possesses genuine multipartite entanglement (GME).
\end{lem}

\begin{proof} The state $\rho_1$ reduces to\begin{equation*}\rho_1 = \frac{1}{2} \left( \proj{0}^{\otimes n} + \proj{1}^{\otimes n} \right),\end{equation*}which is an explicit classical mixture of product states and is therefore fully separable across all partitions. For the second mixture, $\rho_2$ is a GHZ-diagonal state. {\renewcommand{\arraystretch}{0.7}\setlength{\arraycolsep}{2pt}
\begin{equation*}
\rho_2=\frac14
\begin{pmatrix}
1 & \cdots & 0 & \cdots & 0 & \cdots & 1\\
\vdots & & \vdots & & \vdots & & \vdots\\
0 & \cdots & 1 & \cdots & 1 & \cdots & 0\\
\vdots & & \vdots & & \vdots & & \vdots\\
0 & \cdots & 1 & \cdots & 1 & \cdots & 0\\
\vdots & & \vdots & & \vdots & & \vdots\\
1 & \cdots & 0 & \cdots & 0 & \cdots & 1
\end{pmatrix}_{2^n\times 2^n}.
\end{equation*}} For any $n$-qubit GHZ-diagonal state $\rho$, the criterion of G"uhne and Seevinck~\cite{GuhneSeevinck2010} states that $\rho$ is biseparable (free of GME) if and only if\begin{equation}|\rho_{1,2^n}|
\leq
\sum_{j=2}^{2^{n-1}}\sqrt{\rho_{j,j}\,\rho_{2^n+1-j,\,2^n+1-j}}\label{eq:GS1}\end{equation} A violation of Eq.~\eqref{eq:GS1} certifies GME, whereas satisfaction guarantees biseparability. For $\rho_1$ both sides of Eq.~\eqref{eq:GS1} vanish identically ($0 \le 0$). For $\rho_2$ with any $\mathbf{r} \neq \emptyset$, the outer coherence yields $\vert{}\rho_{1,2^n}\vert{} = \frac{1}{4}$. The only nonvanishing term in the right-hand sum corresponds to the complementary bit-strings $x_{\mathbf{r}}$ and $\bar{x}_{\mathbf{r}}$ (indexed by $a, b$ with $a + b = 2^n + 1$), where $\rho_{a,a} = \rho_{b,b} = \frac{1}{4}$. The sum therefore evaluates identically to $\sqrt{\rho_{a,a}\,\rho_{b,b}} = \frac{1}{4}$, satisfying the bound with equality. Thus, neither mixture possesses genuine multipartite entanglement.
\end{proof}

Although neither $\rho_1$ nor $\rho_2$ exhibits GME, the auxiliary resources required for their nondemolition discrimination are strictly inequivalent. The constituent states of $\rho_2$ possess a nonzero relative Hamming distance, enabling their perfect LOCC discrimination while preserving their GHZ entanglement using only a single two-qubit Bell state (Theorem\ref{th_2GHZ_nonzero}). In contrast, the constituent states of $\rho_1$ differ solely in their relative phase $q$ (zero Hamming distance), precluding their entanglement-preserving discrimination under unassisted LOCC (Theorem~\ref{th_GHZ_2}). Preserving the target correlation here necessitates an additional $n$-qubit GHZ state, an auxiliary resource strictly more powerful than a bipartite Bell state and no extraction of resource is allowed.

\section{Conclusion} 

In this Letter, we introduced an operational framework for entanglement-preserving quantum state discrimination under LOCC. By enforcing strict resource invariance-demanding classical information retrieval without degrading target nonclassical correlations, we established the conditions governing catalytic state identification across bipartite and multipartite systems. While net resource extraction is strictly precluded for Bell-state ensembles, multipartite systems exhibit a marked operational asymmetry governed by relative Hamming geometry. 

Specifically, a single shared Bell state serves as a genuine catalyst, enabling the deterministic identification of any orthogonal $n$-qubit GHZ pair separated by a nonzero Hamming distance while leaving the target states intact. This advantage collapses at zero Hamming distance, where auxiliary entanglement offers no operational leverage. Extending this analysis to larger ensembles, we demonstrated that pairwise nonzero Hamming distances guarantee local distinguishability, whereas particular triples remain fundamentally indistinguishable. Furthermore, we uncovered a nontrivial hierarchy of entanglement costs for discriminating larger multipartite ensembles. Broadly, our findings demonstrate that nondestructive state discrimination serves as a sensitive probe of catalytic resource dynamics, revealing nonlocal structures that remain obscured when entanglement consumption is permitted.

\section{Acknowledgment}
R. Rahaman acknowledges support from the ANRF Advanced Research Grant (ARG), Grant No. ANRF/ARG/2025/012066/MS. A.B. acknowledges the financial support from ANRF-NPDF (File no. PDF/2025/002762).
\bibliography{Ref}
\end{document}